\documentclass{article}
\usepackage{microtype,xstring}
\usepackage{amssymb,amsmath,amsthm}
\usepackage{graphics,graphicx,color}
\usepackage{enumerate,paralist,hyperref}
\usepackage{caption,subcaption,xspace}
\usepackage{todonotes}
\usepackage{setspace}
\usepackage{wrapfig}
\usepackage[english]{babel}
\usepackage[a4paper]{geometry}

\title{Universal Slope Sets for 1-Bend Planar Drawings}

\author{Patrizio~Angelini$^1$ \and Michael~A.~Bekos$^1$	\and Giuseppe~Liotta$^2$ \and Fabrizio~Montecchiani$^2$
\\
\medskip
\\
$^1$Wilhelm-Schickhard-Institut f\"ur Informatik, Universit\"at T\"ubingen, Germany\\
\texttt{\{angelini,bekos\}@informatik.uni-tuebingen.de}
\\
$^2$Universit\'a degli Studi di Perugia, Italy\\
\texttt{\{giuseppe.liotta, fabrizio.montecchiani\}@unipg.it}
}

\newcommand{\myparagraph}[1]{\medskip\noindent\textbf{#1}}

\newcommand{\skel}[1]{G^\textit{skel}_{#1}}

\newcommand{\reduced}[1]{\overline{G}_{#1}\xspace}
\newcommand{\nice}{stretchable\xspace}

\newtheorem{theorem}{Theorem}
\newtheorem{lemma}{Lemma}

\newcommand{\topfirst}[2]{%
    \IfEqCase{#1}{%
        {a}{\tau_{a}(#2)}%
        {c}{\tau_{c}(#2)}%
        {~}{\tau(#2)}%
    }[\PackageError{firsttop}{Undefined option to firsttop: #1}{}]%
}%

\newcommand{\bottomfirst}[2]{%
    \IfEqCase{#1}{%
        {a}{\beta_{a}(#2)}%
        {c}{\beta_{c}(#2)}%
        {~}{\beta(#2)}%
    }[\PackageError{bottomfirst}{Undefined option to bottomfirst: #1}{}]%
}

\newcommand{\bottomport}[2]{\beta_{#1}(#2)}

\begin{document}

\maketitle

\begin{abstract}
We describe a set of $\Delta -1$ slopes that are universal for 1-bend planar drawings of planar graphs of maximum degree $\Delta \geq 4$; this establishes a new upper bound of $\Delta-1$ on the 1-bend planar slope number. By universal we mean that every planar graph of degree $\Delta$ has a planar drawing with at most one bend per edge and such that the slopes of the segments forming the edges belong to the given set of slopes. 
This improves over previous results in two ways: Firstly, the best previously known upper bound for the 1-bend planar slope number was $\frac{3}{2} (\Delta -1)$ (the known lower bound being $\frac{3}{4} (\Delta -1)$); 
secondly, all the known algorithms to construct 1-bend planar drawings with $O(\Delta)$ slopes use a different set of slopes for each graph and can have bad angular resolution, while our algorithm uses a universal set of slopes, which also guarantees that the minimum angle between any two edges incident to a vertex is $\frac{\pi}{(\Delta-1)}$.
\end{abstract}

\section{Introduction}

This paper is concerned with planar drawings of graphs such that each edge is a poly-line with few bends, each segment has one of a limited set of possible slopes, and the drawing has good angular resolution, i.e. it forms large angles between consecutive edges incident on a common vertex. Besides their theoretical interest, visualizations with these properties find applications in software engineering and information visualization (see, e.g., \cite{DBLP:books/ph/BattistaETT99,DBLP:books/sp/Juenger04,DBLP:reference/crc/2013gd}). For example, planar graphs of maximum degree four (degree-4 planar graphs) are widely used in database design, where they are typically represented by orthogonal drawings, i.e. crossing-free drawings such that every edge segment is a polygonal chain of horizontal and vertical segments. Clearly, orthogonal drawings of degree-4 planar graphs are optimal both in terms of angular resolution and in terms of number of distinct slopes for the edges.
Also, a classical result in the graph drawing literature is that every degree-4 planar graph, except the octahedron, admits an orthogonal drawing with at most two bends per edge~\cite{DBLP:journals/comgeo/BiedlK98}.

It is immediate to see that more than two slopes are needed in a planar drawing of a graph with vertex degree $\Delta \geq 5$. The {\em k-bend planar slope number} of a graph $G$ with degree $\Delta$ is defined as the minimum number of distinct slopes that are sufficient to compute a crossing-free drawing  of $G$ with at most $k$ bends per edge. 
 Keszegh et al.~\cite{DBLP:journals/siamdm/KeszeghPP13} generalize the aforementioned technique by Biedl and Kant~\cite{DBLP:journals/comgeo/BiedlK98} and prove that for any $\Delta \geq 5$, the 2-bend planar slope number of a degree-$\Delta$ planar graph is $\lceil \Delta/2 \rceil$; the construction in their proof has optimal angular resolution, that is $\frac{2 \pi}{\Delta}$.
 
For the case of drawings with one bend per edge, Keszegh et al.~\cite{DBLP:journals/siamdm/KeszeghPP13} also show an upper bound of $2 \Delta$ and a lower bound of $\frac{3}{4}(\Delta - 1)$ on the 1-bend planar slope number, while a recent paper by Knauer and Walczak~\cite{DBLP:conf/latin/KnauerW16} improves the upper bound to $\frac{3}{2}(\Delta - 1)$. Both these papers use a similar technique: First, the graph is realized as a contact representation with $T$-shapes~\cite{DBLP:journals/cpc/FraysseixMR94}, which is then transformed into a planar drawing where vertices are points and edges are poly-lines with at most one bend. The set of slopes depends on the initial contact representation and may change from graph to graph; also, each slope is either very close to horizontal or very close to vertical, which in general gives rise to bad angular resolution. Note that Knauer and Walczak~\cite{DBLP:conf/latin/KnauerW16} also considered subclasses of planar graphs. In particular, they proved that the 1-bend planar slope number of outerplanar graphs with $\Delta > 2$ is $\lceil \frac{\Delta}{2} \rceil$ and presented an upper bound of $\Delta + 1$ for planar bipartite graphs. 

In this paper, we study the trade-off between number of slopes, angular resolution, and number of bends per edge in a planar drawing of a graph having maximum degree $\Delta$. We improve the upper bound of Knauer and Walczak~\cite{DBLP:conf/latin/KnauerW16} on the 1-bend planar slope number of planar graphs and at the same time we achieve $\Omega(\frac{1}{\Delta})$ angular resolution.  More precisely, we prove the following.

\begin{theorem}\label{th:result}
For any $\Delta \geq 4$, there exists an equispaced universal set $S$ of $\Delta-1$ slopes for 1-bend planar drawings of planar graphs with maximum degree $\Delta$. That is, every such graph has a planar drawing with the following properties: (i) each edge has at most one bend; (ii) each edge segment uses one of the slopes in $S$; and (iii) the minimum angle between any two consecutive edge segments incident on a vertex or a bend is at least $\frac{\pi}{\Delta - 1}$.
\end{theorem}

Theorem~\ref{th:result}, in conjuction with~\cite{DBLP:conf/focs/Kant92}, implies that the 1-bend planar slope number of planar graphs with $n \geq 5$ vertices and maximum degree $\Delta \geq 3$ is at most $\Delta - 1$. We prove the theorem by using an approach that is conceptually different from that of Knauer and Walczak~\cite{DBLP:conf/latin/KnauerW16}: We do not construct an intermediate representation and then transform it into a 1-bend planar drawing, but we prove the existence of a \emph{universal} set of slopes and use it to directly compute a 1-bend planar drawing of any graph with degree at most $\Delta$. The universal set of slopes consists of $\Delta-1$ distinct slopes such that the minimum angle between any two of them is $\frac{\pi}{(\Delta-1)}$. An immediate consequence of the $\frac{3}{4}(\Delta - 1)$ lower bound argument in~\cite{DBLP:journals/siamdm/KeszeghPP13} is that a 1-bend planar drawing with the minimum number of slopes cannot have angular resolution larger than $\frac{4}{3}\frac{\pi}{(\Delta - 1)}$. Hence, the angular resolution of our drawings is optimal up to a multiplicative factor of at most $0.75$; also, note that the angular resolution of a graph of degree $\Delta$ is at most $\frac{2\pi}{\Delta}$ even when the number of slopes and the number of bends along the edges are not bounded.

\begin{wrapfigure}{r}{0.35\textwidth} 
\centering
\includegraphics[scale=1]{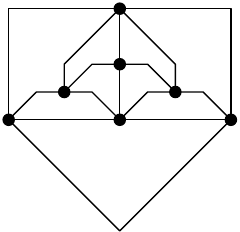}
\caption{A 1-bend planar drawing with 4 slopes and angular resolution $\frac{\pi}{4}$ of a graph with $\Delta=5$.\label{fig:example}}
\end{wrapfigure}
The proof of Theorem~\ref{th:result} is constructive and it gives rise to a linear-time algorithm assuming the real RAM model of computation. Figure~\ref{fig:example} shows a drawing computed with this algorithm. The construction for triconnected planar graphs uses a variant of the shifting method of De Fraysseix, Pach and Pollack~\cite{DBLP:journals/combinatorica/FraysseixPP90}; this construction is the building block for the drawing algorithm for biconnected planar graphs, which is based on the SPQR-tree decomposition of the graph into its triconnected components (see, e.g.,~\cite{DBLP:books/ph/BattistaETT99}). Finally, the result is extended to connected graphs by using a block-cutvertex tree decomposition as a guideline to assign subsets of the universal slope set to the different biconnected components of the input graph. If the graph is disconnected, since we use a universal set of slopes, the distinct connected components can be drawn independently.

\paragraph{Related work.} 
The results on the slope number of graphs are mainly classified into two categories based on whether the input graph is planar or not. For a (planar) graph $G$ of maximum degree~$\Delta$, the \emph{slope number} (\emph{planar slope number}) is the minimum number of slopes that are sufficient to compute a straight-line (planar) drawing of $G$. The slope number of non-planar  graphs is lower bounded by $\lceil \Delta / 2 \rceil$~\cite{DBLP:journals/cj/WadeC94} but it can be arbitrarily large, even when $\Delta = 5$~\cite{DBLP:journals/combinatorics/BaratMW06}. For $\Delta=3$ this number is $4$~\cite{DBLP:conf/gd/MukkamalaP11}, while it is unknown for $\Delta = 4$, to the best of our knowledge. Upper bounds on the slope number are known for complete graphs~\cite{DBLP:journals/cj/WadeC94} and outer $1$-planar graphs~\cite{DBLP:journals/jgaa/GiacomoLM15} (i.e., graphs that can be drawn in the plane such that each edge is crossed at most once, and all vertices are on the external boundary). Deciding whether a graph has slope number $2$ is NP-complete~\cite{DBLP:journals/comgeo/DujmovicESW07,DBLP:journals/siamcomp/FormannHHKLSWW93}.

For a planar graph $G$ of maximum degree $\Delta$, the planar slope number of $G$ is lower bounded by $3\Delta-6$ and upper bounded by $O(2^\Delta)$~\cite{DBLP:journals/siamdm/KeszeghPP13}. Improved upper bounds are known for special subclasses of planar graphs, e.g., planar graphs with $\Delta \leq 3$~\cite{DBLP:journals/comgeo/DujmovicESW07,DBLP:conf/latin/GiacomoLM14,DBLP:conf/wg/Kant92}, outerplanar graphs with $\Delta \geq 4$~\cite{DBLP:journals/comgeo/KnauerMW14}, partial $2$-trees~\cite{DBLP:conf/gd/LenhartLMN13}, planar partial $3$-trees~\cite{DBLP:journals/gc/JelinekJKLTV13}. Note that determining the planar slope number of a graph is hard in the existential theory of the reals~\cite{DBLP:journals/jgaa/Hoffmann17}.

Closely related to our problem is also the problem of finding \emph{$d$-linear} drawings of graphs, in which all angles (that are formed either between consecutive segments of an edge or between edge-segments incident to the same vertex) are multiples of $2 \pi / d$. Bodlaender and Tel~\cite{DBLP:journals/jgaa/BodlaenderT04} showed that, for $d=4$, an angular resolution of $2 \pi / d$ implies $d$-linearity and that this is not true for any $d > 4$.
Special types of $d$-linear drawings are the \emph{orthogonal}~\cite{DBLP:journals/comgeo/BiedlK98,DBLP:journals/comgeo/BlasiusLR16,DBLP:journals/siamcomp/GargT01,DBLP:journals/siamcomp/Tamassia87} and the \emph{octilinear}~\cite{DBLP:journals/jgaa/BekosG0015,DBLP:conf/latin/Bekos0016,Noellenburg05} drawings, for which $d=2$ and $d=4$ holds, respectively. As already recalled, Biedl and Kant~\cite{DBLP:journals/comgeo/BiedlK98}, and independently Liu et al.~\cite{DBLP:journals/dam/LiuMS98}, have shown that any planar graph with $\Delta \leq 4$ (except the octahedron) admits a planar orthogonal drawing with at most two bends per edge. Deciding whether a degree-4 planar graph has an orthogonal drawing with no bends is NP-complete~\cite{DBLP:journals/siamcomp/GargT01}, while it solvable in polynomial time if one bend per edge is allowed (see, e.g.,~\cite{DBLP:journals/algorithmica/BlasiusKRW14}). 
On the other hand, octilinear drawings have been mainly studied in the context of metro map visualization and map schematization~\cite{DBLP:journals/tvcg/NollenburgW11,DBLP:journals/tvcg/StottRMW11}. N{\"o}llenburg~\cite{Noellenburg05} proved that deciding whether a given embedded planar graph with $\Delta \leq 8$ admits a bendless planar octilinear drawing is NP-complete. Bekos et al.~\cite{DBLP:journals/jgaa/BekosG0015} showed that a planar graph with $\Delta \leq 5$ always admits a planar octilinear drawing with at most one bend per edge and that such drawings are not always possible if $\Delta \geq 6$. Note that in our work we generalize their positive result to any $\Delta$. Later, Bekos et al.~\cite{DBLP:conf/latin/Bekos0016} studied bounds on the total number of bends of planar octilinear drawings.

Finally, trade-offs between number of bends, angular resolution, and area requirement of planar drawings of graphs with maximum degree $\Delta$ are, for example, studied in~\cite{DBLP:conf/wg/BonichonSM02,DBLP:journals/dcg/DuncanEGKN13,DBLP:journals/jgaa/DuncanK03,DBLP:conf/gd/DurocherM14,DBLP:conf/esa/GargT94,DBLP:conf/gd/GutwengerM98}.

\paragraph{Paper organization.} The rest of this paper is organized as follows. Preliminaries are given in Section~\ref{sec:preliminaries}. In Section~\ref{sec:triconnected}, we describe a drawing algorithm for triconnected planar graphs.  The technique is extended to biconnected and to general planar graphs in Sections~\ref{sec:biconnected} and~\ref{sec:general}, respectively. Finally, in Section~\ref{sec:conclusions} we discuss further implications of Theorem~\ref{th:result} and we list open problems.

\section{Preliminaries}
\label{sec:preliminaries}

A graph $G=(V,E)$ containing neither loops nor multiple edges is \emph{simple}. We consider simple graphs, if not otherwise specified. The \emph{degree} of a vertex of $G$ is the number of its neighbors. We say that $G$ has \emph{maximum degree} $\Delta$ if it contains a vertex with degree $\Delta$ but no vertex with degree larger than $\Delta$. A graph is \emph{connected}, if for any pair of vertices there is a path connecting them. Graph $G$ is \emph{$k$-connected}, if the removal of $k-1$ vertices leaves the graph connected. A $2$-connected ($3$-connected) graph is also called \emph{biconnected} (\emph{triconnected}, respectively). 

A \emph{drawing} $\Gamma$ of $G$ maps each vertex of $G$ to a point in the plane and each edge of $G$ to a Jordan arc between its two endpoints. A drawing is \emph{planar}, if no two edges cross (except at common endpoints). A planar drawing divides the plane into connected regions, called \emph{faces}. The unbounded one is called \emph{outer face}. A graph is \emph{planar}, if it admits a planar drawing.  A \emph{planar embedding} of a planar graph is an equivalence class of planar drawings that combinatorially define the same set of faces and outer~face.  

The \emph{slope} $s$ of a line $\ell$ is the angle that a horizontal line needs to be rotated counter-clockwise in order to make it overlap with $\ell$.  The slope of an edge-segment is the slope of the line containing the segment. Let $S$ be a set of slopes sorted in increasing order; assume w.l.o.g.\ up to a rotation, that $S$ contains the $0$ angle, which we call \emph{horizontal slope}. A \emph{1-bend} planar drawing $\Gamma$ of graph $G$ \emph{on $S$} is a planar drawing of $G$ in which every edge is composed of at most two straight-line segments, each of which has a slope that belongs to $S$. We say that $S$ is \emph{equispaced} if and only if the difference between any two consecutive slopes of $S$ is $\frac{\pi}{|S|}$. For a vertex $v$ in $G$, each slope $s \in S$ defines two different \emph{rays} that emanate from $v$ and have slope $s$. If $s$ is the horizontal slope, then these rays are called \emph{horizontal}. Otherwise, one of them is the \emph{top} and the other one is the \emph{bottom} ray of $v$. Consider a 1-bend planar drawing $\Gamma$ of a graph $G$ and a ray $r_v$ emanating from a vertex $v$ of $G$. We say that $r_v$ is \emph{free} if there is no edge attached to $v$ through $r_v$. We also say that $r_v$ is \emph{incident} to face $f$ of $\Gamma$ if and only if $r_v$ is free and the first face encountered when moving from $v$ along $r_v$ is $f$. 

Let $\Gamma$ be a 1-bend planar drawing of a graph and let $e$ be an edge incident to the outer face of $\Gamma$ that has a horizontal segment. A \emph{cut at} $e$ is a $y$-monotone curve that 
\begin{inparaenum}[(i)]
\item starts at any point of the horizontal segment of $e$, 
\item ends at any point of a horizontal segment of an edge $e' \neq e$  incident to the outer face of $\Gamma$, and 
\item crosses only horizontal segments of $\Gamma$.
\end{inparaenum}

Central in our approach is the canonical order of triconnected planar graphs~\cite{DBLP:journals/combinatorica/FraysseixPP90,DBLP:journals/algorithmica/Kant96}. Let $G=(V,E)$ be a triconnected planar graph and let $\Pi = (P_0,\ldots,P_m)$ be a partition of $V$ into paths, such that $P_0 = \{v_1,v_2\}$, $P_m=\{v_n\}$, edges $(v_1,v_2)$ and $(v_1,v_n)$ exist and belong to the outer face of $G$. For $k=0,\ldots,m$, let $G_k$ be the subgraph induced by $\cup_{i=0}^k P_i$ and denote by $C_k$  the  outer  face of $G_k$.  $\Pi$ is a \emph{canonical order} of $G$ if for each $k=1,\ldots,m-1$ the following hold: %
\begin{inparaenum}[(i)]
\item $G_k$ is biconnected, 
\item all neighbors of $P_k$ in $G_{k-1}$ are on $C_{k-1}$,
\item $|P_k|=1$ or the degree of each vertex of $P_k$ is two in $G_k$, and  
\item all vertices of $P_k$ with $0\leq k < m$ have at least one neighbor in $P_j$ for some $j > k$.
\end{inparaenum}
A canonical order of any triconnected planar graph can be computed in linear time~\cite{DBLP:journals/algorithmica/Kant96}.

An SPQR-tree $\mathcal{T}$ represents the decomposition of a biconnected graph $G$ into its triconnected components (see, e.g.,~\cite{DBLP:books/ph/BattistaETT99}) and it can be computed in linear time~\cite{DBLP:conf/gd/GutwengerM00}. Every triconnected component of $G$ is associated with a node $\mu$ of  $\mathcal{T}$. The triconnected component itself is called the \emph{skeleton} of $\mu$, denoted by $\skel{\mu}$. A node $\mu$ in $\mathcal{T}$ can be of four different types: %
\begin{inparaenum}[(i)]
\item $\mu$ is an \emph{R-node}, if $\skel{\mu}$ is a triconnected graph,
\item a simple cycle of length at least three classifies $\mu$ as an \emph{S-node},
\item a bundle of at least three parallel edges classifies $\mu$ as a \emph{P-node},
\item the leaves of $\mathcal{T}$ are \emph{Q-nodes}, whose skeleton consists of two parallel edges.
\end{inparaenum}
Neither two $S$- nor two $P$-nodes are adjacent in~$\mathcal{T}$. 
A \emph{virtual edge} in $\skel{\mu}$ corresponds to a tree node $\nu$ that is adjacent to $\mu$ in $\mathcal{T}$, more precisely, to another virtual edge in $\skel{\nu}$. If we assume that $\mathcal{T}$ is rooted at a Q-node $\rho$, then every skeleton (except the one of $\rho$) contains exactly one virtual edge, called \emph{reference edge} and whose endpoints are the \emph{poles} of $\mu$, that has a counterpart in the skeleton of its parent. 
Every subtree $\mathcal{T}_\mu$ rooted at a node $\mu$ of $\mathcal{T}$ induces a subgraph $G_\mu$ of $G$ called \emph{pertinent}, that is described by $\mathcal{T}_\mu$ in the decomposition.

Finally, the \emph{BC-tree} $\mathcal{B}$ of a connected graph $G$ represents the decomposition of $G$ into its biconnected components. $\mathcal{B}$ has a B-node for each biconnected component of $G$ and a C-node for each cutvertex of $G$. Each B-node is connected to the C-nodes that are part of its biconnected~component.
\section{Triconnected Planar Graphs}
\label{sec:triconnected}

Let $G$ be a triconnected planar graph of maximum degree $\Delta \geq 4$ and let $S$ be a set of $\Delta-1$~equispaced slopes containing the horizontal one. We consider the vertices of $G$ according to a canonical order $\Pi = (P_0,\ldots,P_m)$. At each step $k=0,\ldots,m$, we consider the planar graph $G_k^-$ obtained by removing edge $(v_1,v_2)$ from $G_k$. Let $C_k^-$ be the path from $v_1$ to $v_2$ obtained by removing $(v_1,v_2)$ from $C_k$. We seek to construct a $1$-bend planar drawing of $G_k^-$ on $S$ satisfying the following~invariants.%
\begin{enumerate}[{I.}1]
\item \label{i:v1-v2} No part of the drawing lies below vertices $v_1$ and $v_2$, which have the same $y$-coordinate.
\item \label{i:stretch} Every edge on $C_k^-$ has a horizontal segment.   
\item \label{i:rays} Each vertex $v$ on $C_k^-$ has at least as many free top rays incident to the outer face of $G_k^-$ as the number of its neighbors in $G \setminus G_k$.
\end{enumerate}

Once a 1-bend planar drawing on $S$ of $G_m^-$ satisfying Invariants~I.\ref{i:v1-v2}--I.\ref{i:rays} has been constructed, a 1-bend planar drawing on $S$ of $G=G_m^-\cup \{(v_1,v_2)\}$ can be obtained by drawing edge $(v_1,v_2)$ as a polyline composed of two straight-line segments, one attaching at the first clockwise bottom ray of $v_1$ and the other one at the first anti-clockwise bottom ray of $v_2$. Note that, since $S$ has at least three slopes, these two rays cross. Invariant~I.\ref{i:v1-v2} ensures that edge $(v_1,v_2)$ does not introduce any crossing. 
In the following lemma, we show an important property of any $1$-bend planar drawing on $S$ satisfying Invariants~I.\ref{i:v1-v2}--I.\ref{i:rays}. 

\begin{lemma}\label{le:stretching}
Let $\Gamma_k$ be a $1$-bend planar drawing on $S$ of $G_k^-$ satisfying Invariants~I.\ref{i:v1-v2}--I.\ref{i:rays}. Let $(u,v)$ be an edge of $C_k^-$ such that $u$ precedes $v$ along path $C_k^-$ and let $\sigma$ be any positive number. It is possible to construct a $1$-bend planar drawing $\Gamma_k'$ on $S$ of $G_k^-$, satisfying Invariants~I.\ref{i:v1-v2}--I.\ref{i:rays}, in which the horizontal distance between any two consecutive vertices along $C_k^-$ is the same as in $\Gamma_k$, except for $u$ and $v$, whose horizontal distance is increased by~$\sigma$.
\end{lemma}
\begin{proof}
We first show that there exists a cut of $\Gamma_k$ at $(u,v)$ that separates the subpath of $C_k^-$ connecting $v_1$ to $u$ from the subpath of $C_k^-$ connecting $v$ to $v_2$. We  use this cut to construct $\Gamma_k'$ as a copy of $\Gamma_k$ in which all the horizontal segments that are crossed by the cut are elongated by~$\sigma$.
	
By Invariant I.\ref{i:stretch}, edge $(u,v)$ has a horizontal segment, which is the first segment crossed by the cut we are going to construct. Then, consider the internal face $f$ of $\Gamma_k$ edge $(u,v)$ is incident to; this face is uniquely defined since $G_k^-$ is biconnected and $(u,v)$ is incident to the outer face. By the properties of the canonical order, there exists at least an edge $e_f$ incident to $f$ that belongs to $G_{k-1}^-$ but not to $G_k^-$; in particular, this edge belongs to $C_{k-1}^-$, and hence has a horizontal segment, by Invariant I.\ref{i:stretch}; see Fig.~\ref{fig:stretching-before}. We thus make our cut traverse face $f$ and cross the horizontal segment of $e_f$. By repeating this argument until reaching the outer face, we obtain the desired cut.  

\begin{figure}[b]
	\centering 
	\begin{minipage}[b]{.3\textwidth}
		\centering 
		\includegraphics[width=.95\textwidth, page=1]{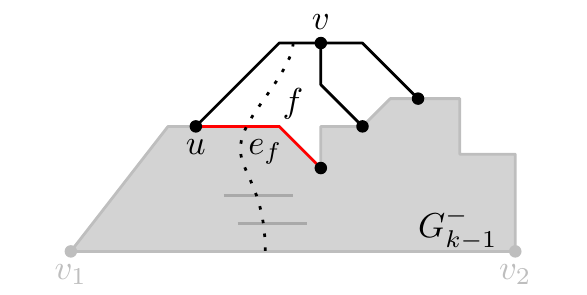}
		\subcaption{~}\label{fig:stretching-before}
	\end{minipage}
	\hfil
	\begin{minipage}[b]{.3\textwidth}
		\centering 
		\includegraphics[width=.95\textwidth, page=2]{img/stretching}
		\subcaption{~}\label{fig:stretching-after}
	\end{minipage}
	\hfil
	\begin{minipage}[b]{.3\textwidth}
		\centering 
		\includegraphics[width=.95\textwidth,page=3]{img/stretching}
		\subcaption{~}\label{fig:shooting-to-inf}
	\end{minipage}
	\caption{Illustrations for (a-b) Lemma~\ref{le:stretching} and (c) Lemma~\ref{le:shooting-to-inf}.}
	\label{fig:stretch-lemmas}
\end{figure}

We now describe how to obtain a drawing $\Gamma_k'$ of $G_k^-$ satisfying all the required properties; refer to Fig.~\ref{fig:stretching-after}. Let $L$ and $R$ be the two sets of vertices separated by the cut. All the vertices in $L$ and all the edges between any two of them are drawn in $\Gamma_k'$ as in $\Gamma_k$; all the vertices in $R$ and all the edges between any two of them are drawn in $\Gamma_k'$ as in $\Gamma_k$, after a translation to the right by $\sigma$. Finally, for each edge that is crossed by the cut, the part that is not horizontal, if any, is drawn in $\Gamma_k'$ as in $\Gamma_k$, while the horizontal part is elongated by $\sigma$.

We prove that $\Gamma_k'$ satisfies all the required properties. First, $\Gamma_k'$ is a 1-bend planar drawing of $G_k^-$ on $S$ since $\Gamma_k$ is. Invariant I.\ref{i:v1-v2} holds since the $y$-coordinates of the vertices have not been changed, while Invariants~I.\ref{i:stretch}--I.\ref{i:rays} hold since all the edges are attached to their incident vertices in $\Gamma_k'$ using the same rays as in $\Gamma_k$. The fact that the horizontal distances among consecutive vertices of $C_k^-$ are the required ones descends from the fact that $L$ contains all the vertices in the path of $C_k^-$ from $v_1$ to $u$, while $R$ contains all the vertices in the path of $C_k^-$ from $v$ to $v_2$. 
\end{proof}

Invariant~I.\ref{i:rays} guarantees that every vertex on $C_k^-$ has enough free top rays incident to the outer face to attach all its incident edges following it in the canonical order. The next lemma shows that these rays can be always used to actually draw these edges (see Fig.~\ref{fig:shooting-to-inf}).

\begin{lemma}\label{le:shooting-to-inf}
Let $\Gamma_k$ be a $1$-bend planar drawing on $S$ of $G_k^-$ satisfying Invariants~I.\ref{i:v1-v2}--I.\ref{i:rays}. Let $u$ be any vertex of $C_k^-$, and let $r_u$ be any free top ray of $u$ that is incident to the outer face of $G_k^-$ in $\Gamma_k$. Then, it is possible to construct a $1$-bend planar drawing $\Gamma_k'$ on $S$ of $G_k^-$, satisfying Invariants~I.\ref{i:v1-v2}--I.\ref{i:rays}, in which $r_u$ does not cross any edge of $\Gamma_k'$.
\end{lemma}
\begin{proof}
Since $r_u$ is a top ray of $u$ incident to the outer face of $\Gamma_k^-$ and due to Invariant I.\ref{i:v1-v2}, if $r_u$ crosses some edges of $G_k^-$, then at least one of these belongs to $C_k^-$. So, we can focus on removing the crossings with the edges of $C_k^-$. Let $P_1$ be the path of $C_k^-$ between $v_1$ and $u$, and let $P_2$ be the path of $C_k^-$ between $u$ and $v_2$. Also, let $u_1$ and $u_2$ be the neighbors of $u$ in $P_1$ and $P_2$, respectively. 
Refer to Fig.~\ref{fig:shooting-to-inf}. By Lemma~\ref{le:stretching}, we can elongate $(u,u_1)$ to eliminate all crossings between $r_u$ and edges of $P_1$ without introducing any new crossings between $r_u$ and edges of $P_2$. We also elongate $(u,u_2)$ to eliminate all crossings between $r_u$ and edges of $P_2$ without introducing any new crossings between $r_u$ and edges of $P_1$. The obtained drawing $\Gamma_k'$ satisfies all the requirements of the lemma. This concludes our proof.
\end{proof}

We now describe our algorithm. First, we draw $P_0=\{v_1,v_2\}$ and $P_1=\{v_3,\ldots,v_j\}$ of partition $\Pi$ such that $v_1, v_3, \dots, v_j, v_2$ lie along a horizontal line, in this order (recall that edge $(v_1,v_2)$ is not considered). Invariants~I.\ref{i:v1-v2} and I.\ref{i:stretch} clearly hold. Invariant I.\ref{i:rays}  follows from the fact that $S$ contains $\Delta-2$ top rays and all vertices drawn so far (including $v_1$ and $v_2$) have at most $\Delta-2$ neighbors later in the canonical order. We now describe how to add path $P_k$, for some $k > 1$, to a drawing $\Gamma_{k-1}$ satisfying Invariants I.\ref{i:v1-v2}--I.\ref{i:rays}, in such a way that the resulting drawing $\Gamma_{k}$ of $G_{k}^-$ is a 1-bend planar drawing on $S$ satisfying Invariants~I.\ref{i:v1-v2}--I.\ref{i:rays}. We distinguish two cases, based on whether $P_k$ is a chain or a singleton.

Suppose first that $P_k$ is a chain, say $\{v_i,v_{i+1},\ldots,v_j\}$; refer to Fig.~\ref{fig:chain}. Let $u_\ell$ and $u_r$ be the neighbors of $v_i$ and $v_j$ in $C_{k-1}^-$, respectively. 
By Invariant I.\ref{i:rays}, each of $u_\ell$ and $u_r$ has at least one free top ray that is incident to the outer face of $\Gamma_{k-1}$; among them, we denote by $\topfirst{a}{u_\ell}$ the first one in anti-clockwise order for $u_\ell$, and by $\topfirst{c}{u_r}$ the first one in clockwise order for $u_r$. By Lemma~\ref{le:shooting-to-inf}, we can assume that $\topfirst{a}{u_\ell}$ and $\topfirst{c}{u_r}$ do not cross any edge in $\Gamma_{k-1}$. This implies that there exists a horizontal line-segment $h$ whose left and right endpoints are on $\topfirst{a}{u_\ell}$ and $\topfirst{c}{u_r}$, respectively, that does not cross any edge of $\Gamma_{k-1}$. We place all the vertices $v_i,v_{i+1},\ldots,v_j$ of $P_k$ on interior points of $h$, in this left-to-right order. Then, we draw edge $(u_\ell,v_i)$ with a segment along $h$ and the other one along $\topfirst{a}{u_\ell}$; we draw edge $(u_r,v_j)$ with a segment along $h$ and the other one along $\topfirst{c}{u_r}$, and we draw every edge $(v_q,v_{q+1})$, with $q=i,\dots,j-1$, with a unique segment along $h$. 

\begin{figure}[t]
    \centering 
    \begin{minipage}[b]{.45\textwidth}
        \centering
        \includegraphics[width=\textwidth]{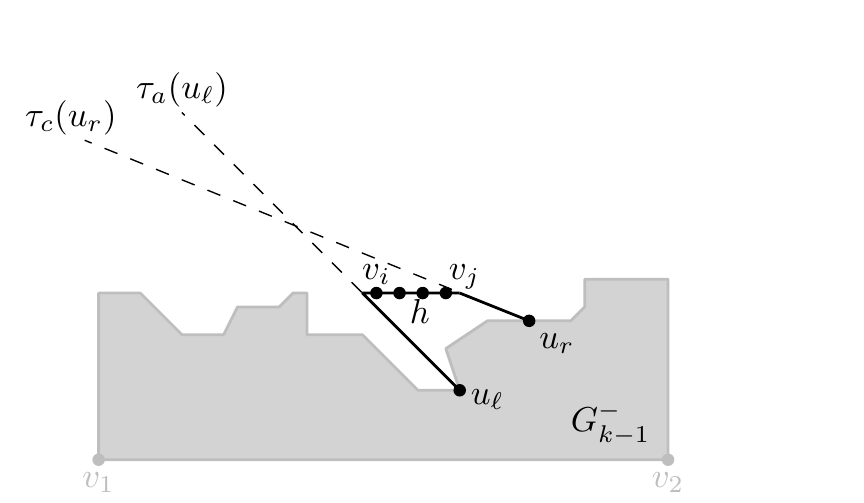}
        \subcaption{~}\label{fig:chain}{}
    \end{minipage}
	\hfil
    \begin{minipage}[b]{.45\textwidth}
        \centering
        \includegraphics[width=\textwidth]{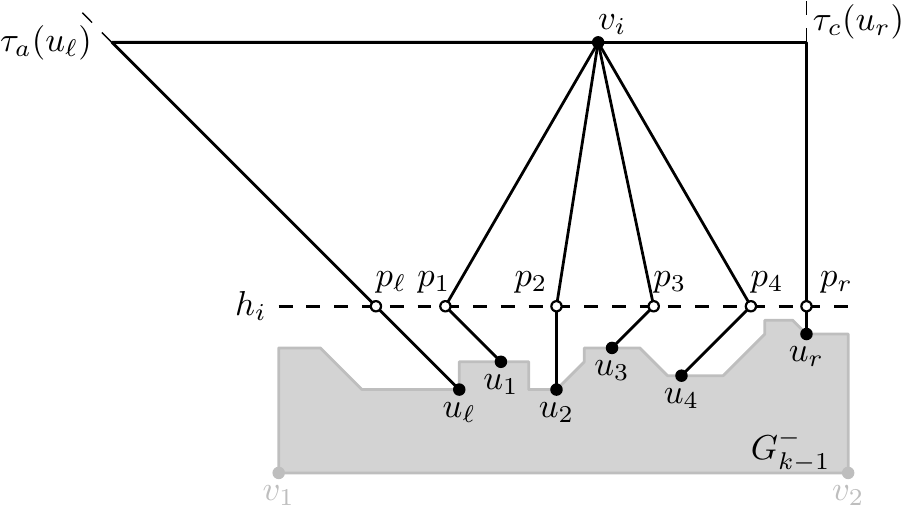}
        \subcaption{~}\label{fig:singleton}{}
    \end{minipage}
    \caption{Illustration of the cases of: 
    (a) a chain, 
    (b) a singleton of degree $\delta_i$ in $G_k$.}
    \label{fig:4p_ref}
\end{figure}

By construction, $\Gamma_k$ is a planar drawing on $S$.  
All the vertices of $P_k$ lie above $u_\ell$ and $u_r$, since $\topfirst{a}{u_\ell}$ and $\topfirst{c}{u_r}$ are top rays of $u_\ell$ and $u_r$, respectively. Hence, these vertices and their incident edges lie above $v_1$ and $v_2$, and thus Invariant I.\ref{i:v1-v2} is satisfied by $\Gamma_k$. Invariant I.\ref{i:stretch} is satisfied since every edge that is drawn at this step has a segment along $h$, which is horizontal. Invariant~I.\ref{i:rays} is satisfied since we attached edges $(u_\ell, v_i)$ and $(u_r,v_j)$ at vertices $u_\ell$ and $u_r$ using the first anti-clockwise free top ray of $u_\ell$ and the first clockwise free top ray of $u_r$ among those incident to the outer face, respectively. Thus, we reduced only by one the number of free top rays incident to the outer face for $u_\ell$ and $u_r$. For the other vertices of $P_k$, the invariant is satisfied since their $\Delta-2$ top rays are free and incident to the outer face. This concludes our description for the case in which $P_k$ is a chain.

Suppose now that $P_k$ is a singleton, say $\{v_i\}$, of degree $\delta_i \leq \Delta$ in $G_{k}^-$. This also includes the case in which $k=m$, that is, $P_k$ is the last path of $\Pi$. If $\delta_i=2$, then $v_i$ is placed as in the case of a chain. So, we may assume in the following that $\delta_i \geq 3$. Let $u_\ell,u_1,u_2,\ldots,u_{\delta_i-2},u_r$ be the neighbors of $v_i$ as they appear along $C_{k-1}^-$.

Refer to Fig.~\ref{fig:singleton}. By Invariant I.\ref{i:rays}, each neighbor of $v_i$ in $C_{k-1}^-$ has at least one free top ray that is incident to the outer face of $\Gamma_{k-1}$; among them, we denote by $\topfirst{a}{u_\ell}$ the first one in anti-clockwise order for $u_\ell$ and by $\topfirst{c}{u_r}$ the first one in clockwise order for $u_r$, as in the case in which $P_k$ is a chain, while for each vertex $u_q$, with $q=1,\dots,\delta_i-2$, we denote by $\topfirst{~}{u_q}$ any of these rays arbitrarily. By Lemma~\ref{le:shooting-to-inf}, we can assume that these rays do not cross any edge in $\Gamma_{k-1}$.

Consider any horizontal line $h_i$ lying above all vertices of $\Gamma_{k-1}$. Rays $\topfirst{a}{u_\ell}$, $\topfirst{~}{u_1}, \dots, \topfirst{~}{u_{\delta_i-2}}$, $\topfirst{c}{u_r}$ cross $h_i$; however, the corresponding intersection points $p_\ell$, $p_1, \dots, p_{\delta_i-2}$, $p_r$ may not appear in this left-to-right order along $h_i$; see Fig.~\ref{fig:intersections-before}. To guarantee this property, we perform a sequence of stretchings of $\Gamma_{k-1}$ by repediately applying Lemma~\ref{le:stretching}. First, if $p_\ell$ is not the leftmost of these intersection points, let $\sigma$ be the distance between $p_\ell$ and the leftmost intersection point. We apply Lemma~\ref{le:stretching} on any edge between $u_\ell$ and $u_1$ along $C_{k-1}^-$ to stretch $\Gamma_{k-1}$ so that all the vertices in the path of $C_k^-$ from $u_1$ to $v_2$ are moved to the right by a quantity $\sigma'$ slightly larger than $\sigma$. This implies that $p_\ell$ is not moved, while all the other intersection points are moved to the right by a quantity $\sigma'$, and thus they all lie to the right of $p_\ell$ in the new drawing; see Fig.~\ref{fig:intersections-after}. Analogously, we can move $p_1$ to the left of every other intersection point, except for $p_\ell$, by applying Lemma~\ref{le:stretching} on any edge between $u_1$ and $u_2$ along $C_{k-1}^-$. Repeating this argument allows us to assume that in $\Gamma_{k-1}$ all the intersection points appear in the correct left-to-right order along $h_i$.

\begin{figure}[t]
	\centering 
	\begin{minipage}[b]{.45\textwidth}
		\centering
		\includegraphics[width=\textwidth,page=1]{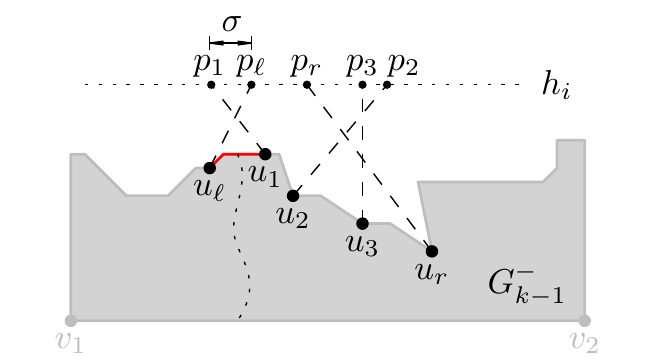}
		\subcaption{~}\label{fig:intersections-before}{}
	\end{minipage}
	\hfil
	\begin{minipage}[b]{.45\textwidth}
		\centering 
		{\includegraphics[width=\textwidth,page=2]{img/intersections}}
		\subcaption{~}\label{fig:intersections-after}{}
	\end{minipage}
	\caption{(a) Intersection points $p_\ell,p_1, \dots, p_{\delta_i-2},p_r$ appear in a wrong order along $h_i$. 
	(b) Applying Lemma~\ref{le:stretching} to make $p_\ell$ be the leftmost intersection point.}
	\label{fig:intersections}
\end{figure}

We now describe the placement of $v_i$. Let $\bottomport{1}{v_i},\ldots,\bottomport{\delta_i-2}{v_i}$ be any set of ${\delta_i-2}$ consecutive bottom rays of $v_i$; to see that $v_i$ has enough bottom rays, recall that $S$ contains $\Delta-1$ slopes and that $\delta_i \leq \Delta$. Observe that, if we place $v_i$ above $h_i$, rays $\bottomport{1}{v_i},\bottomport{2}{v_i},\ldots,\bottomport{\delta_i-2}{v_i}$ intersect $h_i$ in this left-to-right order. Let $\rho_1,\dots,\rho_{\delta_i-2}$ be the corresponding intersection points. The goal is to place $v_i$ so that each $\rho_q$, with $q=1,\dots,\delta_i-2$, coincides with $p_q$. To do so, consider the line $\lambda_1$ passing through $p_1$ with the same slope as $\bottomport{1}{v_i}$. Observe that placing $v_i$ on $\lambda_1$ above $h_i$ results in $\rho_1$ to coincide with $p_1$. Also note that, while moving $v_i$ upwards along $\lambda_1$, the distance $d(\rho_q,\rho_{q+1})$ between any two consecutive points $\rho_q$ and $\rho_{q+1}$, with $q=1,\dots,\delta_i-3$, increases. 

We move $v_i$ upwards along $\lambda_1$ in such a way that $d(\rho_q,\rho_{q+1}) > d(p_1,p_{\delta_i-2})$, for each $q=1,\dots,\delta_i-3$. This implies that all points $p_2,\dots,p_{\delta_i-2}$ lie strictly between $\rho_1$ and $\rho_2$. Then, we apply Lemma~\ref{le:stretching} on any edge between $u_1$ and $u_2$ along $C_{k-1}^-$ to stretch $\Gamma_{k-1}$ so that all the vertices in the path of $C_k^-$ from $u_2$ to $v_2$ are moved to the right by a quantity $d(p_2,\rho_2)$. In this way, $u_1$ is not moved and so $p_1$ still coincides with $\rho_1$; also, $p_2$ is moved to the right to coincide with $\rho_2$; finally, since $d(\rho_2,\rho_3) > d(p_1,p_{\delta_i-2}) > d(p_2,p_{\delta_i-2})$, all points $p_3,\dots,p_{\delta_i-2}$ lie strictly between $\rho_2$ and $\rho_3$. By repeating this transformation for all points $p_3,\dots,p_{\delta_i-2}$, if any, we guarantee that each $\rho_q$, with $q=1,\dots,\delta_i-2$, coincides with $p_q$. We draw each edge $(v_i,u_q)$, with $q=1,\dots,\delta_i-2$, with a segment along $\topfirst{~}{u_q}$ and the other one along $\bottomport{q}{v_i}$. 

It remains to draw edges $(v_i,u_\ell)$ and $(v_i,u_r)$, which by Invariant I.\ref{i:stretch} must have a horizontal segment. After possibly applying Lemma~\ref{le:stretching} on any edge between $u_\ell$ and $u_1$ along $C_{k-1}^-$ to stretch $\Gamma_{k-1}$, we can guarantee that $\topfirst{a}{u_\ell}$ crosses the horizontal line through $v_i$ to the left of $v_i$. Similarly, we can guarantee that $\topfirst{c}{u_r}$ crosses the horizontal line through $v_i$ to the right of $v_i$ by applying Lemma~\ref{le:stretching} on any edge between $u_{\delta_i-2}$ and $u_r$. We draw edge $(v_i,u_\ell)$ with one segment along $\topfirst{a}{u_\ell}$ and one along the horizontal line through $v_i$, and we draw edge $(v_i,u_r)$ with one segment along $\topfirst{c}{u_r}$ and one along the horizontal line through $v_i$. A drawing produced by this algorithm is illustrated in Fig.~\ref{fig:singleton}.

The fact that $\Gamma_k$ is a 1-bend planar drawing on $S$ satisfying Invariant I.\ref{i:v1-v2}--I.\ref{i:rays} can be shown as for the case in which $P_k$ is a chain. In particular, for Invariants I.\ref{i:stretch} and I.\ref{i:rays}, note that vertices $u_1,\dots,u_{\delta_i-2}$ do not have neighbors in $G \setminus G_k$ and do not belong to $C_k^-$. Thus, they do not need to have any free top ray incident to the outer face of $G_k^-$ and the edges connecting them to $v_i$ do not need to have a horizontal segment. This concludes our description for the case in which $P_k$ is a singleton, and yields the following theorem.

\begin{theorem}\label{th:triconnected}
For any $\Delta \geq 4$, there exists a equispaced universal set $S$ of $\Delta-1$ slopes for 1-bend planar drawings of triconnected planar graphs with maximum degree $\Delta$. 
Also, for any such graph on $n$ vertices, a 1-bend planar drawing on $S$ can be computed in $O(n)$ time.
\end{theorem}
\begin{proof}
Apply the algorithm described above to produce a 1-bend planar drawing of $G$ on $S$. The  correctness has been proved through out the section.
We now prove the time complexity. As already mentioned, computing the canonical order $\Pi$ of $G$ takes linear time~\cite{DBLP:journals/algorithmica/Kant96}. Hence, our algorithm can be easily implemented in quadratic time. In fact, when a chain is added, we apply Lemma~\ref{le:stretching} a constant number of times. For a singleton $v_i$ of degree $\delta_i \leq \Delta$, instead, we may apply this lemma $O(\delta_i)$ times. However, since $\sum_{i=1}^n \delta_i = O(n)$, the total number of applications of the lemma over all singletons is $O(n)$. The total quadratic time descends from the fact that a straightforward application of Lemma~\ref{le:stretching} may require linear time.
To improve the time complexity of our algorithm to linear we seek to use the shifting method of Kant~\cite{DBLP:conf/focs/Kant92}. However, as the $y$-coordinates of the vertices are not consecutive, this method is~not directly applicable. On the other hand, observe that the $y$-coordinates of the vertices that have been placed at some step of our algorithm do not change in later steps. As noted by Bekos et al.~\cite{DBLP:journals/jgaa/BekosG0015}, one can exploit this observation so to allow the usage of the shifting method (even in the case of non-consecutive $y$-coordinates) in order to perform all applications of Lemma~\ref{le:stretching} in total linear~time.
\end{proof}
\section{Biconnected Planar Graphs}
\label{sec:biconnected}

In this section we describe how to extend Theorem~\ref{th:triconnected} to biconnected planar graphs, using the SPQR-tree data structure described in Section~\ref{sec:preliminaries}.

The idea is to traverse the SPQR-tree of the input biconnected planar graph $G$ bottom-up and to construct for each visited node a drawing of its pertinent graph (except for its two poles) inside a rectangle, which we call \emph{chip}. Besides being a 1-bend planar drawing on $S$, this drawing must have an additional property, namely that it is possible to increase its width while changing neither its height nor the slope of any edge-segment. We call this property \emph{horizontal stretchability}. In the following, we give a formal definition of this drawing and describe how to compute it for each type of node of the SPQR-tree.

Let $\mathcal{T}$ be the SPQR-tree of $G$ rooted at an arbitrary Q-node $\rho$. Let $\mu$ be a node of $\mathcal{T}$ with poles $s_\mu$ and $t_\mu$. Let $G_\mu$ be the pertinent graph of $\mu$. Let $\reduced{\mu}$ be the graph obtained from $G_\mu$ as follows. First, remove edge $(s_\mu,t_\mu)$, if it exists; then, subdivide each edge incident to $s_\mu$ (to $t_\mu$) with a dummy vertex, which is a \emph{pin of $s_\mu$} (is a \emph{pin of $t_\mu$}); finally, remove $s_\mu$ and $t_\mu$, and their incident edges. Note that, if $\mu$ is a Q-node other than the root $\rho$, then $\reduced{\mu}$ is the empty graph. We denote by $\delta(s_\mu,\mu)$ and $\delta(t_\mu,\mu)$ the degree of $s_\mu$ and $t_\mu$ in $G_\mu$, respectively; note that the number of pins of $s_\mu$ (of $t_\mu$) is $\delta(s_\mu,\mu)-1$ (is $\delta(t_\mu,\mu)-1$), if edge $(s_\mu,t_\mu)$ exists in $G$, otherwise it is $\delta(s_\mu,\mu)$ (it is $\delta(t_\mu,\mu)$).

The goal is to construct a 1-bend planar drawing of $\reduced{\mu}$ on $S$ that lies inside an axis-aligned rectangle, called the \emph{chip} of $\mu$ and denoted by $C(\mu)$, so that the following invariant properties are satisfied (see Figure~\ref{fig:chip}):
\begin{enumerate}[{P.}1:]
\item \label{p:pins} All the pins of $s_\mu$ lie on the left side of $C(\mu)$, while all the pins of $t_\mu$ lie on its right side;
\item \label{p:edge} for each pin, the unique edge incident to it is horizontal; and
\item \label{p:side} there exist pins on the bottom-left and on the bottom-right corners of $C(\mu)$.
\end{enumerate}

We call \emph{horizontally-stretchable} (or \emph{\nice}, for short) a drawing of $\reduced{\mu}$ satisfying Properties P.\ref{p:pins}-P.\ref{p:side}. Note that a \nice drawing $\Gamma$ remains \nice after any uniform scaling, any translation, and any horizontal or vertical flip, since the horizontal slope is in $S$ and the slopes are equispaced. On the other hand, it is generally not possible to perform any non-uniform scaling of $\Gamma$ (in particular, a horizontal or a vertical scaling) without altering the slopes of some segments. However, we can simulate a horizontal scaling up of $\Gamma$ by elongating the horizontal segments incident to all the pins lying on the same vertical side of the chip, thus obtaining a new \nice drawing inside a new chip with the same height and a larger width. Conversely, a horizontal scaling down cannot always be simulated in this way.

Before giving the details of the algorithm, we describe a subroutine that we will often use to add the poles of a node $\mu$ to a \nice drawing of $\reduced{\mu}$ and draw the edges incident to them.

\begin{figure}[t]
	\centering 
    \begin{minipage}[b]{.37\textwidth}
        \centering
        \includegraphics[width=\textwidth,page=1]{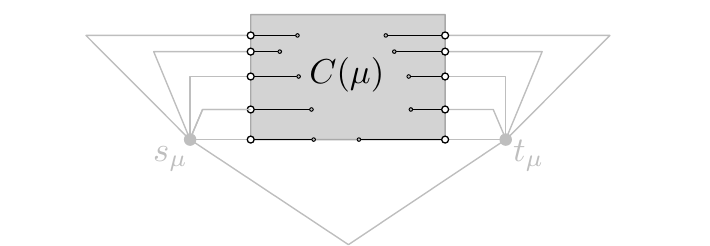}
        \subcaption{~}\label{fig:chip}{} 
    \end{minipage} 
	\hfil
	\begin{minipage}[b]{.37\textwidth}
		\centering 
		\includegraphics[width=\textwidth,page=4]{img/chip}
		\subcaption{~}\label{fig:draw-edges-pole}{}
	\end{minipage}
	\caption{Illustrations (a) of a pin $C(\mu)$ and (b) for Lemma~\ref{le:draw-edges-pole}.}
\end{figure}

\begin{lemma}\label{le:draw-edges-pole}
Let $u \in \{s_{\mu},t_{\mu}\}$ be a pole of a node $\mu \in \mathcal{T}$ and let $u_1,\dots,u_q$ be $q$ neighbors of $u$ in $\reduced{\mu}$. Consider a \nice drawing of $\reduced{\mu}$ inside a chip $C(\mu)$, whose pins $p_1,\dots,p_q$ correspond to $u_1,\dots,u_q$. Suppose that there exists a set of $q$ consecutive free rays of $u$ and that the elongation of the edge incident to each pin $p_1,\dots,p_q$ intersects all these rays. Then, it is possible to draw edges $(u,u_1),\dots,(u,u_q)$ with two straight-line segments whose slopes are in $S$, without introducing any crossing between two edges incident to $u$ or between an edge incident to $u$ and an edge of $\reduced{\mu}$.
\end{lemma}
\begin{proof}
Refer to Fig.~\ref{fig:draw-edges-pole}. First note that, since $p_1,\dots,p_q$ are all on the same side of $C(\mu)$, the elongations of their incident edges intersect the $q$ free rays of $u$ in the same order; we name the rays as $r_1,\dots,r_q$ according to this order. Also note that, since the elongations of the edges incident to all the pins intersect all of $r_1,\dots,r_q$, the elongation of the edge incident to either $p_1$ or $p_q$ separates $u$ from all the other pins. We assume w.l.o.g.\ that the elongation of the edge incident to $p_1$ separates $u$ from $p_2,\dots,p_q$, as in Fig.~\ref{fig:draw-edges-pole}. We then place each pin $p_i$, with $1 \leq i \leq q$, on the intersection point between the elongation of its incident edge and $r_i$, and draw edge $(u,u_i)$ as a poly-line with a single bend at $p_i$. This procedure yields indeed a drawing satisfying the required properties by construction and by the fact that the drawing of $\reduced{\mu}$ is \nice. 
\end{proof}

We now describe the algorithm. At each step of the bottom-up traversal of $\mathcal{T}$, we consider a node $\mu \in \mathcal{T}$ with children $\nu_1, \dots, \nu_h$, and we construct a \nice drawing of $\reduced{\mu}$ inside a chip $C(\mu)$ starting from the \nice drawings of $\reduced{\nu_1}, \dots, \reduced{\nu_h}$ inside chips $C(\nu_1),\dots, C(\nu_h)$ that have been already constructed. In the following, we distinguish four different cases, according to which $\mu$ is a Q-, a P-, an S-, or an R-node.

\myparagraph{Suppose that $\mu$ is a Q-node}. If $\mu$ is not the root $\rho$ of $\mathcal{T}$, we do not do anything, since $\reduced{\mu}$ is the empty graph; edge $(s_\mu,t_\mu)$ of $G$ corresponding to $\mu$ will be drawn when visiting either the parent $\xi$ of $\mu$, if $\xi$ is not a P-node, or the parent of $\xi$.
In the case in which $\mu$ is indeed the root $\rho$ of $\mathcal{T}$, that is $\mu=\rho$, we observe that it has only one child $\nu_1$. Since $\reduced{\mu}$ coincides with $\reduced{\nu_1}$, the \nice drawing of $\reduced{\nu_1}$ is also a \nice drawing of $\reduced{\mu}$. Vertices $s_\mu$ and $t_\mu$, and their incident edges, will be added at the end of the traversal of $\mathcal{T}$.

\myparagraph{Suppose that $\mu$ is a P-node}; refer to Fig.~\ref{fig:p}. We consider a chip $C(\mu)$ for $\mu$ whose height is larger than the sum of the heights of chips $C(\nu_1), \dots, C(\nu_h)$ and whose width is larger than the one of any of $C(\nu_1), \dots, C(\nu_h)$. 
Then, we place chips $C(\nu_1), \dots, C(\nu_h)$ inside $C(\mu)$ so that no two chips overlap, their left sides lie along the left side of $C(\mu)$, and the bottom side of $C(\nu_h)$ lies along the bottom side of $C(\mu)$. Finally, we elongate the edges incident to all the pins on the right side of $C(\nu_1), \dots, C(\nu_h)$ till reaching the right side of $C(\mu)$. The resulting drawing is \nice since each of the drawings of $\reduced{\nu_1}, \dots, \reduced{\nu_h}$ is \nice. In particular, Property P.\ref{p:side} holds for $C(\mu)$ since it holds for $C(\nu_h)$.

\begin{figure}[t]
    \centering 
    \begin{minipage}[b]{.2\textwidth}
        \centering 
        \includegraphics[width=\textwidth,page=2]{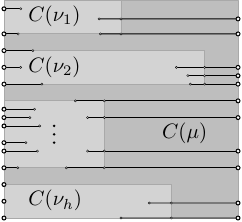}
        \subcaption{~}\label{fig:p}
    \end{minipage}
    \hfil
    \begin{minipage}[b]{.7\textwidth}
        \centering
        \includegraphics[width=\textwidth,page=2]{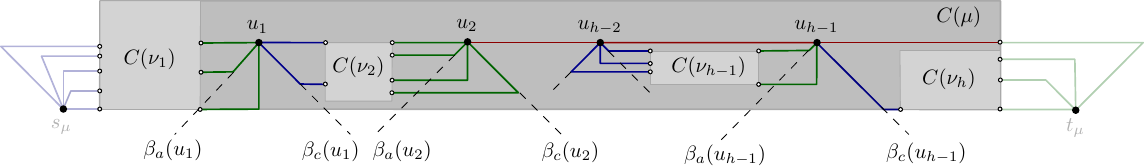}
        \subcaption{~}\label{fig:s}
    \end{minipage}
    \caption{Illustrations for the cases in which $\mu$ is: 
    (a) a $P$-node and (b) an $S$-node.}
\end{figure}

\myparagraph{Suppose that $\mu$ is an S-node}; refer to Fig.~\ref{fig:s}. Let $u_1, \dots, u_{h-1}$ be the internal vertices of the path of virtual edges between $s_\mu$ and $t_\mu$ that is obtained by removing the virtual edge $(s_\mu,t_\mu)$ from the skeleton of $\mu$. 
We construct a \nice drawing of $\reduced{\mu}$ as follows. 

First, we place vertices $u_1, \dots, u_{h-1}$ in this order along a horizontal line $l_\mu$. For $i=1, \dots, h-1$, let $\bottomfirst{a}{u_i}$ and $\bottomfirst{c}{u_i}$ be the first bottom rays of $u_i$ in anti-clockwise and in clockwise order, respectively. To place each chip $C(\nu_i)$, with $i=2,\dots, h-1$, we first flip it vertically, so that it has pins on its top-left and top-right corners, by Property~P.\ref{p:side}. After possibly scaling it down uniformly, we place it in such a way that its left side is to the right of $u_{i-1}$, its right side is to the left of $u_i$, it does not cross $\bottomfirst{c}{u_{i-1}}$ and $\bottomfirst{a}{u_i}$, and either its top side lies on line $l_\mu$ (if edge $(u_{i-1},u_{i}) \notin G$; see $C(\nu_2)$ in Fig.~\ref{fig:s}), or it lies slightly below it (otherwise; see $C(\nu_{h-1})$ in Fig.~\ref{fig:s}).

Then, we place $C(\nu_1)$ and $C(\nu_h)$, after possibly scaling them up uniformly, in such a way that:
\begin{inparaenum}[(i)]
	\item Chip $C(\nu_1)$ lies to the left of $u_1$ and does not cross $\bottomfirst{a}{u_1}$. Also, if $(s_\mu,u_1) \in G$, then $C(\nu_1)$ lies entirely below $l_\mu$; otherwise, as in Fig.~\ref{fig:s}, the topmost pin on its right side has the same $y$-coordinate as $u_1$.
	\item Chip $C(\nu_h)$ lies to the right of $u_h$ and does not cross $\bottomfirst{c}{u_h}$. Also, if $(u_h,t_\mu) \in G$, as in Fig.~\ref{fig:s}, then $C(\nu_h)$ lies entirely below $u_h$; otherwise, the topmost pin on its left side has the same $y$-coordinate as $u_h$.
	\item The bottom sides of $C(\nu_1)$ and of $C(\nu_h)$ have the same $y$-coordinate, which is smaller than the one of the bottom side of any other chip $C(\nu_2), \dots, C(\nu_{h-1})$.
\end{inparaenum}

We now draw all the edges incident to each vertex $u_i$, with $i=1,\dots,h-1$. 
If edge $(u_{i-1},u_i) \in G$, then it can be drawn as a horizontal segment, by construction. Otherwise, $u_i$ can be connected with a horizontal segment to its neighbor in $\reduced{\nu_i}$ corresponding to the topmost pin on the right side of $C(\nu_i)$. In both cases, one of these edges is attached at a horizontal ray of $u_i$. Analogously, one of the edges connecting $u_i$ to its neighbors in $\reduced{\nu_{i+1}} \cup \{u_{i+1}\}$ is attached at the other horizontal ray of $u_i$. Thus, it is possible to draw the remaining $\delta(u_i,\nu_i)+\delta(u_i,\nu_{i+1})-2 \leq \Delta-2$ edges incident to $u_i$ by attaching them at the $\Delta-2$ bottom rays of $u_i$, by applying Lemma~\ref{le:draw-edges-pole}. In fact, since $C(\nu_i)$ and $C(\nu_{i+1})$ lie to the left and to the right of $u_i$, respectively, and do not cross $\bottomfirst{c}{u_i}$ and $\bottomfirst{a}{u_i}$, the elongations of the edges incident to the pins of $u_i$ in $C(\nu_i)$ and in $C(\nu_{i+1})$ corresponding to these edges intersect all the bottom rays of $u_i$, hence satisfying the preconditions to apply the lemma.

Finally, we construct chip $C(\mu)$ as the smallest rectangle enclosing all the current drawing. Note that the left side of $C(\mu)$ contains the left side of $C(\nu_1)$, while the right side of $C(\mu)$ contains the right side of $C(\nu_h)$. Thus, all the pins of $s_\mu$, possibly except for the one corresponding to edge $(s_\mu,u_1)$, lie on the left side of $C(\mu)$. Also, if $(s_\mu,u_1)$ exists, we can add the corresponding pin since, by construction, $C(\nu_1)$ lies entirely below $u_1$. The same discussion applies for the pins of $t_\mu$. This proves that the constructed drawing satisfies properties P.\ref{p:pins} and P.\ref{p:edge}. To see that it also satisfies P.\ref{p:side}, note that the bottom side of $C(\mu)$ contains the bottom sides of $C(\nu_1)$ and of $C(\nu_h)$, by construction, which have a pin on both corners, by Property~P.\ref{p:side}. Thus, the constructed drawing of $\reduced{\mu}$ is \nice.

\myparagraph{Suppose that $\mu$ is an R-node}. We compute a \nice drawing of $\reduced{\mu}$ as follows. First, we compute a 1-bend planar drawing on $S$ of the whole pertinent $G_\mu$ of $\mu$, including its poles $s_\mu$ and $t_\mu$; then, we remove the poles of $\mu$ and their incident edges, we define chip $C(\mu)$, and we place the pins on its two vertical sides so to satisfy Properties P.\ref{p:pins}--P.\ref{p:side}.

In order to compute the drawing of $G_\mu$, we exploit the fact that the skeleton $\skel{\mu}$ of $\mu$ is triconnected. Hence, we can use the algorithm described in Section~\ref{sec:triconnected} as a main tool for drawing $G_\mu$, with suitable modifications to take into account the fact that each virtual edge $(u,v)$ of $\skel{\mu}$ actually corresponds to a whole subgraph, namely the pertinent graph $G_\nu$ of the child $\nu$ of $\mu$ with poles $s_\nu=u$ and $t_\nu=v$. Thus, when the virtual edge $(u,v)$ is considered, we have to add the \nice drawing of $\reduced{\nu}$ inside a chip $C(\nu)$; this enforces additional requirements for our drawing algorithm. 

The first obvious requirement is that $(u,v)$ will occupy $\delta(u,\nu)$ consecutive rays of $u$ and $\delta(v,\nu)$ consecutive rays of $v$, and not just a single ray for each of them, as in the triconnected case. However, reserving the correct amount of rays of $u$ and $v$ is not always sufficient to add $C(\nu)$ and to draw the edges between $u$, $v$, and vertices in $\reduced{\nu}$. In fact, we need to ensure that there exists a placement for $C(\nu)$ such that the elongations of the edges incident to its pins intersect all the reserved rays of the poles $u$ and $v$ of $\nu$, hence satisfying the preconditions to apply Lemma~\ref{le:draw-edges-pole}. In a high-level description, for the virtual edges that would be drawn with a horizontal segment in the triconnected case (all the edges of a chain, and the first and last edges of a singleton), this can be done by using a construction similar to the one of the case in which $\mu$ is an S-node. For the edges that do not have any horizontal segment (the internal edges of a singleton), instead, we need a more complicated construction.

We now describe the algorithm, which is again based on considering the vertices of $H=G_\mu$ according to a canonical order $\Pi = (P_0,\ldots,P_m)$ of $H$, in which $v_1=s_\mu$ and $v_2=t_\mu$, and on constructing a 1-bend planar drawing of $H_k^-$ on $S$ satisfying a modified version of Invariants~I.\ref{i:v1-v2}--I.\ref{i:rays}. 

\begin{enumerate}[{M.}1]
	\item \label{m:v1-v2} No part of the drawing lies below vertices $v_1$ and $v_2$, which have the same $y$-coordinate.
	\item \label{m:stretch} For every virtual edge $(w,z)$ on $C_k^-$, if $(w,z)$ belongs to $H$ then it has a horizontal segment; also, the edge-segments corresponding to edges incident to the pins of the chip of the child of $\mu$ corresponding to $(w,z)$ are horizontal.
	\item \label{m:rays} Each vertex $v$ on $C_k^-$ has at least as many free top rays incident to the outer face of $H_k^-$ as the number of its neighbors in $H$ that have not been drawn yet.
\end{enumerate}

We note that Invariant~M.\ref{m:v1-v2} is identical to Invariant~I.\ref{i:v1-v2}, while Invariant~M.\ref{m:rays} is the natural extension of Invariant~I.\ref{i:rays} to take into account our previous observation. Finally, Invariant~M.\ref{m:stretch} corresponds to Invariant~I.\ref{i:stretch}, as it ensures that we can still apply Lemma~\ref{le:stretching} and Lemma~\ref{le:shooting-to-inf}. 

At the first step, we draw $P_0=\{v_1,v_2\}$ and $P_1=\{v_3,\ldots,v_j\}$. Consider the path of virtual edges $(v_1,v_3),(v_3,v_4),\dots,(v_j,v_2)$. Let $\nu_{1,3}, \nu_{3,4}, \dots,\nu_{j,2}$ be the corresponding children of $\mu$, and let $C(\nu_{1,3}), C(\nu_{3,4}), \dots,C(\nu_{j,2})$ be their chips. 
We consider this path as the skeleton of an S-node with poles $v_1$ and $v_2$, and we we apply the same algorithm as in the case in which $\mu$ is an S-node to draw the subgraph composed of $v_3,\ldots,v_j$ and of chips $C(\nu_{1,3}), C(\nu_{3,4}),\dots,C(\nu_{j,2})$ inside a larger chip, denoted by $C(v_1,v_2)$. Note that, by construction, $C(v_1,v_2)$ has pins on its bottom-left and on its bottom-right corners. We then place $v_1$ and $v_2$ with the same $y$-coordinate as the bottom side of $C(v_1,v_2)$, with $v_1$ to the left and $v_2$ to the right of $C(v_1,v_2)$. We draw one of the edges incident to $v_1$ horizontal, and the remaining $\delta(v_1,\nu_{1,3})-1$ by applying Lemma~\ref{le:draw-edges-pole}, and the same for $v_2$.
Invariants~M.\ref{m:v1-v2} and M.\ref{m:stretch} are satisfied by construction. For Invariant~M.\ref{m:rays}, note that $v_3,\ldots,v_j$ have all their $\Delta-2$ top rays free, by construction, and at least two of their neighbors have already been drawn. Also, $v_1$ and $v_2$ have consumed only $\delta(v_1,\nu_{1,3})-1$ and $\delta(v_2,\nu_{j,2})-1$ top rays, respectively. Since edge $(v_1,v_2)$ does not belong to $H_k^-$ (but belongs to $H$), $v_1$ and $v_2$ satisfy Invariant~M.\ref{m:rays}.

We now describe how to add path $P_k$, for some $k > 1$, to the current drawing $\Gamma_{k-1}$ in the two cases in which $P_k$ is a chain or a singleton.

Suppose that $P_k$ is a chain, say $\{v_i,v_{i+1},\ldots,v_j\}$; let $u_\ell$ and $u_r$ be the neighbors of $v_i$ and $v_j$ in $C_k^-$. Let $\nu_\ell,\nu_i,\dots,\nu_{j-1},\nu_r$ be the children of $\mu$ corresponding to virtual edges $(u_\ell,v_i),(v_i,v_{i+1}),$ $\dots,$ $(v_{j-1},v_j),(v_j,v_r)$, and let $C(\nu_\ell),C(\nu_i),\dots,C(\nu_{j-1}),C(\nu_r)$ be their chips. 

We define rays $\topfirst{a}{u_\ell}$ and $\topfirst{c}{u_r}$, and the horizontal segment $h$ between them, as in the triconnected case. Due to Lemma~\ref{le:shooting-to-inf}, we can assume that $\topfirst{a}{u_\ell}$ and the $\delta(u_\ell,\nu_\ell)-1$ top rays of $u_\ell$ following it in anti-clockwise order do not cross any edge of $\Gamma_{k-1}$, and the same for $\topfirst{c}{u_r}$ and the $\delta(u_r,\nu_r)-1$ top rays of $u_r$ following it in clockwise order. Note that, by Invariant~M.\ref{m:rays}, all these rays are free. As in the step in which we considered $P_0$ and $P_1$ of $\Pi$, we use the algorithm for the case in which $\mu$ is an S-node to construct a drawing of the subgraph composed of $v_i,\ldots,v_j$ and of chips $C(\nu_\ell), C(\nu_i), \dots,C(\nu_{j-1}),C(\nu_r)$ inside a larger chip $C(u_\ell,u_r)$, which has pins on its bottom-left and on its bottom-right corners. We then place $C(u_\ell,u_r)$ so that its bottom side lies on $h$ and it does not cross $\topfirst{a}{u_\ell}$ and $\topfirst{c}{u_r}$, after possibly scaling it down uniformly. Finally, we draw the $\delta(u_\ell,\nu_\ell)$ edges between $u_\ell$ and its neighbors in $\reduced{\nu_\ell} \cup \{v_i\}$, and the $\delta(u_r,\nu_r)$ edges between $u_r$ and its neighbors in $\reduced{\nu_r} \cup \{v_j\}$, by applying Lemma~\ref{le:draw-edges-pole}, whose preconditions are satisfied. The fact that the constructed drawing satisfies the three invariants can be proved as in the previous case.

Suppose finally that $P_k$ is a singleton, say $\{v_i\}$, of degree $\delta_i \leq \Delta$ in $H_{k}^-$. As in the triconnected case, we shall assume that $\delta_i \geq 3$. Let $u_\ell,u_1,\ldots,u_{\delta_i-2},u_r$ be the neighbors of $v_i$ as they appear along $C_{k-1}^-$, let $\nu_\ell, \nu_1,\dots,\nu_{\delta_i-2},\nu_r$ be the children of $\mu$ corresponding to the virtual edges connecting $v_i$ with these vertices, and let $C(\nu_\ell), C(\nu_1),\dots,C(\nu_{\delta_i-2}),C(\nu_r)$ be their chips.

For each $q=1,\dots,\delta_i-2$, we select any set $T_q$ of consecutive $\delta(u_q,\nu_q)$ free top rays of $u_q$ incident to the outer face and a set $B_q$ of consecutive $\delta(v_i,\nu_q)$ bottom rays of $v_i$; see Fig.~\ref{fig:sing2}. Sets $B_1,\dots,B_{\delta_i-2}$ are selected in such a way that all the rays in $B_q$ precede all the rays in $B_{q+1}$ in anti-clockwise order. Since $\delta(v_i,\nu_\ell)+\delta(v_i,\nu_r) \geq 2$, vertex $v_i$ has enough bottom rays for sets $B_1,\dots,B_{\delta_i-2}$. We also define sets $T_\ell$ and $T_r$ as composed of the first $\delta(u_\ell,\nu_\ell)$ free top rays of $u_\ell$ in anti-clockwise order and of the first $\delta(u_r,\nu_r)$ free top rays of $u_r$ in clockwise order, respectively.

We then select a horizontal line $h_i$ lying above every vertex in $\Gamma_{k-1}$. As in the algorithm described in Section~\ref{sec:triconnected}, after possibly applying $O(\Delta)$ times Lemma~\ref{le:stretching}, we can assume that all the rays in sets $T_\ell,T_1,\dots,T_{\delta_i-2},T_r$ intersect $h_i$ in the correct order. Namely, when moving along $h_i$ from left to right, we encounter all the intersections with the rays in $T_\ell$, then all those with the rays in $T_1$, and so on. On the other hand, this property is already guaranteed for the rays in $B_1,\dots,B_{\delta_i-2}$. This defines two total left-to-right orders $\mathcal{O}_T$ and $\mathcal{O}_B$ of the intersection points of $T_\ell,T_1,\dots,T_{\delta_i-2},T_r$ and of $B_1,\dots,B_{\delta_i-2}$ along $h_i$, respectively. To simplify the description, we extend these orders to the rays in $T_\ell,T_1,\dots,T_{\delta_i-2},T_r$ and in $B_1,\dots,B_{\delta_i-2}$, respectively.

Our goal is to merge the two sets of intersection points, while respecting $\mathcal{O}_T$ and $\mathcal{O}_B$, in such a way that the following condition holds for each $q=1,\dots,\delta_i-2$. If edge $(v_i,u_q)$ belongs to $H$, then the first intersection point of $T_q$ in $\mathcal{O}_T$ coincides with the first intersection point of $B_q$ in $\mathcal{O}_B$, and the second intersection point of $T_q$ in $\mathcal{O}_T$ is to the right of the last intersection point of $B_q$ in $\mathcal{O}_B$; see $T_1$ and $B_1$ in Fig.~\ref{fig:sing2}. Otherwise, $(v_i,u_q) \notin H$ and the first intersection point of $T_q$ in $\mathcal{O}_T$ is to the right of the last intersection point of $B_q$ in $\mathcal{O}_B$; see $T_3$ and $B_3$ in Fig.~\ref{fig:sing2}. In both cases, the intersection points of $T_q$ and $B_q$ are to the left of those of $T_{q+1}$ and $B_{q+1}$.

To obtain this goal, we perform a procedure analogous to the one described in Section~\ref{sec:triconnected} to make points $p_1,\dots,p_{\delta_i-2}$ coincide with points $\rho_1,\dots,\rho_{\delta_i-2}$. Namely, we consider a line $\lambda_1$, whose slope is the one of the first ray in $B_1$, that starts at the first intersection point of $T_1$ in $\mathcal{O}_T$, if edge $(v_1,u_1)$ belongs to $H$, or at any point between the last intersection point of $T_1$ and the first intersection point of $T_2$ in $\mathcal{O}_T$, otherwise. Then, we place $v_i$ along $\lambda_1$, far enough from $h_i$ so that the distance between any two consecutive intersection points in $\mathcal{O}_B$ is larger than the distance between the first and the last intersection points in $\mathcal{O}_T$; see Fig.~\ref{fig:sing1}. Finally, we apply Lemma~\ref{le:stretching} at most $\delta_i-3$ times to move the intersection points of sets $T_2,\dots,T_{\delta_i-2}$, one by one, in the their correct positions; see Fig.~\ref{fig:sing2}.

\begin{figure}[t]
    \centering 
    \begin{minipage}[b]{.35\textwidth}
    	\begin{minipage}[b]{\textwidth}
    		\includegraphics[width=\textwidth,page=3]{img/biconnected}
    		\subcaption{~}\label{fig:sing1}
    	\end{minipage}    	
    	\begin{minipage}[b]{\textwidth}
    		\includegraphics[width=\textwidth,page=4]{img/biconnected}
    		\subcaption{~}\label{fig:sing2}
    	\end{minipage}    	    	
    \end{minipage}
    \hfil
    \begin{minipage}[b]{.48\textwidth}
        \includegraphics[width=.9\textwidth,page=1]{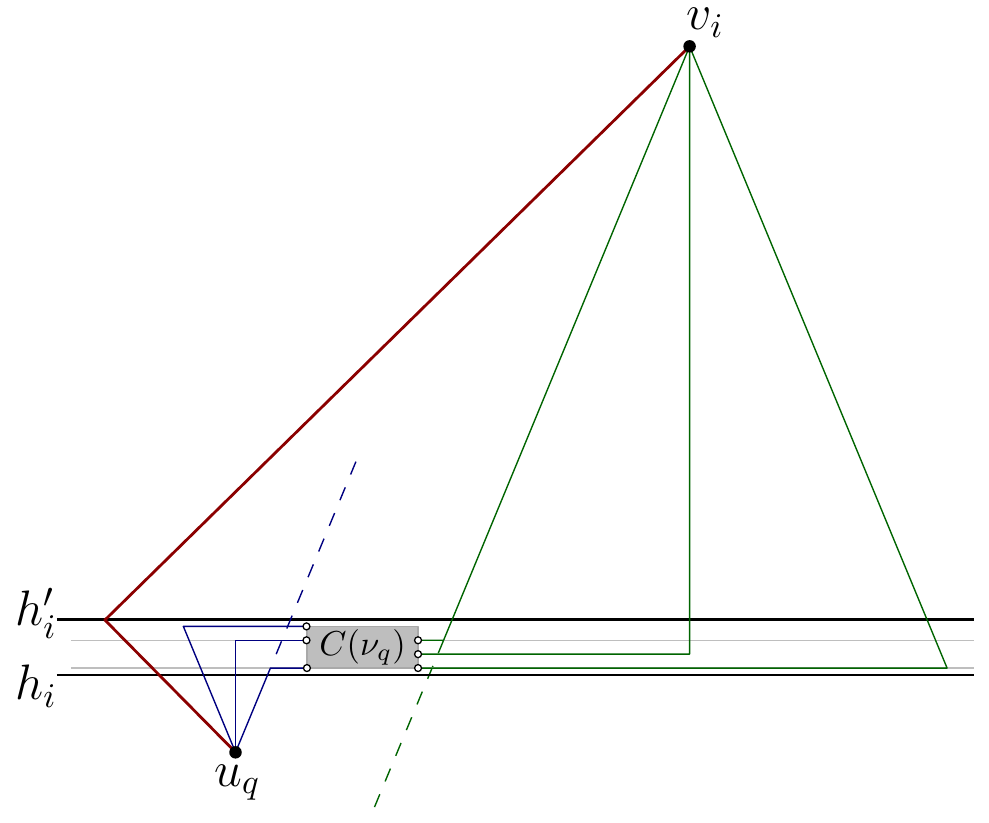}
        \subcaption{~}\label{fig:vedge}
    \end{minipage}
    \caption{Ilustrations for placing singleton $v_i$ in the case of an R-node.}
    \label{fig:spampanato}
\end{figure}

Once the required ordering of intersection points along $h_i$ has been obtained, we consider another horizontal line $h'_i$ lying above $h_i$ and close enough to it so that its intersections with the rays in $T_\ell,T_1,\dots,T_{\delta_i-2},T_r$ and  $B_1,\dots,B_{\delta_i-2}$ appear along it in the same order as along $h_i$. We place each chip $C(\nu_q)$, with $q=1,\dots,\delta_i-2$, after possibly scaling it down uniformly, in the interior of the region delimited by these two lines, by the last ray in $T_q$, and by a ray in $B_q$ (either the second or the first, depending on whether $(v_i,u_q) \in H$ or not), so that its top side is horizontal; see Fig.~\ref{fig:vedge}.

We draw the edges incident to $v_i$ and $u_q$, for each $q=1,\dots,\delta_i-2$, as follows. If edge $(v_i,u_q)$ belongs to $H$, we draw it with one segment along the first ray in $T_q$ and one along the first ray in $B_q$ (see the red edge in Fig.~\ref{fig:vedge}). For the other edges we apply Lemma~\ref{le:draw-edges-pole} twice, whose preconditions are satisfied due to the placement of $C(\nu_q)$ (see the blue and green edges in Fig.~\ref{fig:vedge}). 

We conclude by drawing the edges connecting $v_i$, $u_\ell$, and vertices in $\reduced{\nu_\ell}$; the edges connecting $v_i$, $u_r$ and vertices in $\reduced{\nu_r}$ are drawn symmetrically. First, after possibly applying Lemma~\ref{le:stretching}, we assume that the last ray of $T_\ell$ intersects the horizontal line through $v_i$ to the left of $v_i$, at a point $p_i$. 
After possibly scaling $C(\nu_\ell)$ down uniformly, we place it so that its left side is to the right of $p_i$, its right side is to the left of $v_i$, it does not cross the first top ray of $v_i$ in clockwise order, and its bottom side is horizontal and lies either above the horizontal line through $v_i$, if edge $(u_\ell,v_i)$ belongs to $H$, or along it, otherwise. Then, we draw $(u_\ell,v_i)$, if it belongs to $H$, with one segment along the last ray of $T_\ell$ and the other one along the horizontal line through $v_i$. Otherwise, edge $(u_\ell,v_i)$ does not belong to $H$ and we can draw one of the edges incident to $v_i$ with a horizontal segment. We finally apply Lemma~\ref{le:draw-edges-pole} twice, to draw the edges from $u_\ell$ to its neighbors in $\reduced{\nu_\ell}$, and from $v_i$ to its other neighbors in $\reduced{\nu_\ell}$. The fact that the constructed drawing satisfies Invariants~M.\ref{m:v1-v2}--M.\ref{m:rays} can be proved as in the triconnected case.

Once the last path $P_m$ of $\Pi$ has been added, we have a drawing $\Gamma_\mu$ of $H=G_\mu$ satisfying Invariants~M.\ref{m:v1-v2}--M.\ref{m:rays}. We construct chip $C(\mu)$ as the smallest axis-aligned rectangle enclosing $\Gamma_\mu$. By Invariant~M.\ref{m:v1-v2}, vertices $v_1$ and $v_2$ lie on the bottom side of $C(\mu)$. Also, by Invariant~M.\ref{m:stretch}, all the edges incident to $v_1$ or to $v_2$ have a horizontal segment. Thus, it is possible to obtain a drawing of $\reduced{\mu}$ inside $C(\mu)$ by removing $v_1$ and $v_2$ (and their incident edges) from $\Gamma_\mu$, by elongating the horizontal segments incident to them till reaching the vertical sides of $C(\mu)$, and by placing pins at their ends. The fact that this drawing satisfies Properties~P.\ref{p:pins}--P.\ref{p:side} follows from the observation that $v_1$ and $v_2$ were on the bottom side of $C(\mu)$. This concludes the case in which $\mu$ is an R-node.

Once we have visited the root $\rho$ of $\mathcal{T}$, we have a \nice drawing of $\reduced{\rho}$ inside a chip $C(\rho)$, which we extend to a drawing of $G$ as follows. Refer to Fig.~\ref{fig:chip}. 
We place $s_\rho$ and $t_\rho$ at the same $y$-coordinate as the bottom side of $C(\rho)$, one to its left and one to its right, so that $C(\rho)$ does not cross any of the rays of $s_\rho$ and of $t_\rho$. Then, we draw edge $(s_\rho,t_\rho)$ with one segment along the first bottom ray in clockwise order of $s_\rho$ and the other one along the first bottom ray in anti-clockwise order of $t_\rho$. Also, we draw the edges connecting $s_\rho$ and $t_\rho$ to the vertices corresponding to the lowest pins on the two vertical sides of $C(\rho)$ as horizontal segments. Finally, we draw all the remaining edges incident to $s_\rho$ and $t_\rho$ by applying Lemma~\ref{le:draw-edges-pole} twice.
The following theorem summarizes the discussion in this section.

\begin{theorem}\label{th:biconnected}
For any $\Delta \geq 4$, there exists a equispaced universal set $S$ of $\Delta-1$ slopes for 1-bend planar drawings of biconnected planar graphs with maximum degree $\Delta$. 
Also, for any such graph on $n$ vertices, a 1-bend planar drawing on $S$ can be computed in $O(n)$ time.
\end{theorem}
\begin{proof}
Apply the algorithm described above to produce a 1-bend planar drawing of $G$ on $S$. The correctness has been proved through out the section.
For the time complexity, first observe that the SPQR-tree $\mathcal{T}$ of $G$ can be computed in linear time~\cite{DBLP:conf/gd/GutwengerM00}. Also, for each node $\mu \in \mathcal{T}$, we can compute a \nice drawing of $\reduced{\mu}$ in time linear in the size of $\skel{\mu}$ assuming that, for each chip, we only store the coordinates of two opposite corners. Final coordinates can then be assigned by traversing the SPQR-tree top-down. Also, notice that for R-nodes, a drawing of the skeleton can be obtained in linear time by Theorem~\ref{th:triconnected}. Since the total size over all the skeletons of the nodes of $\mathcal{T}$ is linear in the size of $G$, our algorithm is linear. 
\end{proof}
\section{General Planar Graphs}
\label{sec:general}

Let $G$ be a connected planar graph of maximum degree $\Delta$ and let $\mathcal{B}$ be its BC-tree. We traverse $\mathcal{B}$ bottom-up; at each step, we consider a $B$-node $\beta$, whose parent in $\mathcal{B}$ is the $C$-node $\gamma$. We exploit Theorem~\ref{th:biconnected} to compute a 1-bend planar drawing $\Gamma(\beta)$ of $\beta$ on the slope-set $S$ with $\Delta-1$ equispaced slopes, assuming that the root of the SPQR-tree of $\beta$ corresponds to an edge incident to $\gamma$.  Consider any vertex $c$ of $\beta$ different from $\gamma$ that is a cut-vertex in $G$, and let $\delta_c$ be the degree of $c$ in $G \setminus \beta$. The construction satisfies the following two invariants.

\begin{enumerate}[{K.}1]
\item \label{it:k1} There exists a set of $\delta_c$ consecutive rays of $c$ that are not used in $\Gamma(\beta)$.

\item \label{it:k2} The edges incident to $\gamma$ use a set of consecutive rays in $\Gamma(\beta)$.
\end{enumerate}

Note that K.\ref{it:k2} is already satisfied by the algorithm of Theorem~\ref{th:biconnected}. For K.\ref{it:k1}, we slightly modify this algorithm. The modified algorithm still guarantees K.\ref{it:k2}. Namely, when vertex $c$ is considered in the bottom-up traversal of the SPQR-tree of $\beta$, we reserve $\delta_c<\Delta$ consecutive rays around $c$.

For each $C$-node $c$ that is a child of $\beta$ in $\mathcal{B}$, consider all its children $\zeta_1, \dots, \zeta_q$, with $q \leq \Delta-2$. By Invariants K.\ref{it:k1} and K.\ref{it:k2}, each of these blocks has been drawn so that $c$ is one of its poles (and therefore drawn on its outer face) and its incident edges use a set of consecutive rays. Since the sum of the degrees of $c$ in these blocks is equal to $\delta_c$, we can insert these drawings into $\Gamma(\beta)$ using the $\delta_c$ free rays of $c$. After possibly scaling the drawings of $\zeta_1, \dots, \zeta_q$ down uniformly and by rotating them appropriately, we can guarantee that these insertions do not introduce any crossings between them or with edges of $\Gamma(\beta)$.

Since the visit of $\mathcal{B}$ can be done in linear time and since a drawing of a disconnected graph can be obtained by drawing each connected component independently, the proof of Theorem~\ref{th:result}~follows.
\section{Conclusions and Open Problems}
\label{sec:conclusions}

In this paper, we improved the best-known upper bound of Knauer and Walczak~\cite{DBLP:conf/latin/KnauerW16} on the 1-bend planar slope number from $\frac{3}{2}(\Delta - 1)$ to $\Delta-1$, for $\Delta \geq 4$. Two side-results of our work are the following. Since the angular resolution of our drawings is at least $\frac{\pi}{\Delta - 1}$, at the cost of increased drawing area our main result also improves the best-known upper bound of $\frac{\pi}{4\Delta}$ on the angular resolution of 1-bend poly-line planar drawings by Duncan and Kobourov~\cite{DBLP:journals/jgaa/DuncanK03}. For $\Delta=4$, it also guarantees that planar graphs with maximum degree $4$ admit 1-bend planar drawings on a set of slopes $\{0,\frac{\pi}{3},\frac{2\pi}{3}\}$, while previously it was known that such graphs can be embedded with one bend per edge on a set of slopes $\{0,\frac{\pi}{4},\frac{\pi}{2},\frac{3\pi}{4}\}$~\cite{DBLP:journals/jgaa/BekosG0015} and with two bends per edge on a set of slopes $\{0,\pi\}$~\cite{DBLP:journals/comgeo/BiedlK98}. 

Our work raises several open problems. 
\begin{inparaenum}[(i)]
\item Reduce the gap between the $\frac{3}{4}(\Delta-1)$ lower bound and the $\Delta-1$ upper bound on the 1-bend planar slope number. 
\item Our algorithm produces drawings with large (possibly super-polynomial) area. Is this unavoidable for 1-bend planar drawings with few slopes and good angular resolution? 
\item Study the straight-line case (e.g., for degree-4 graphs). Note that the stretching operation might be difficult in this setting.
\item We proved that a set of $\Delta -1$ equispaced slopes is universal for 1-bend planar drawings. Is \emph{every} set of $\Delta -1$ slopes universal? Note that for $\Delta \leq 4$ a positive answer descends from our work and from a result by Dujmovic et al.~\cite{DBLP:journals/comgeo/DujmovicESW07}, who proved that any planar graph that can be drawn on a particular set of three slopes can also be drawn on any set of three slopes. If the answer to this question is negative for $\Delta > 4$, what is the minimum value $s(\Delta)$ such that every set of $s(\Delta)$ slopes is universal?
\end{inparaenum}

\paragraph{Acknowledgements.}
This work  started at the 19$^{th}$ {\em Korean Workshop on Computational Geometry}. We wish to thank the organizers and the participants for creating a pleasant and stimulating atmosphere and in particular Fabian Lipp and Boris Klemz for useful discussions.

\bibliographystyle{abbrv}
\bibliography{references}

\end{document}